\documentclass[twoside]{article}
\usepackage[a4paper]{geometry}
\usepackage[latin1]{inputenc} 
\usepackage[T1]{fontenc} 
\usepackage{RR}
\usepackage{hyperref}

\usepackage{amsmath,amssymb,amsfonts,amsthm}
\usepackage{mathrsfs}
\usepackage{enumerate}
\usepackage{color}
\usepackage{subcaption}
\newcommand{\RR}{\mathbf{R}}

\usepackage{algorithm}
\usepackage{algpseudocode}
\usepackage[font=small]{caption}
\usepackage[numbers]{natbib}

\newcommand{\argmin}[1]{\underset{#1}{\operatorname{arg}\,\operatorname{min}}\;}
\newcommand{\argmax}[1]{\underset{#1}{\operatorname{arg}\,\operatorname{max}}\;}
\newcommand{\prox}[1]{\mathrm{prox}_{#1}}
\newcommand{\dom}[1]{\mathrm{dom}{(#1)}}
\newcommand{\T}{\mathrm{T}}
\usepackage{color}
\usepackage[dvipsnames]{xcolor}
\usepackage{algorithm}

\usepackage{graphicx}

\newtheorem{theorem}{Theorem}
\newtheorem{fact}{Fact}
\newtheorem{claim}{Proposition}
\theoremstyle{definition}
\newtheorem{definition}{Definition}
\newtheorem*{scheme}{Scheme}

\newtheorem{corollary}{Corollary}
\RRNo{9015}
\RRdate{January 2017}
\RRauthor{
Zaid Allybokus\thanks{Huawei Technologies \& Inria Sophia Antipolis}
\and
Konstantin Avrachenkov\thanks{Inria Sophia Antipolis}
\and
J\'er\'emie Leguay\thanks{Huawei Technologies}
\and
Lorenzo Maggi\thanks{Huawei Technologies}
}

\authorhead{Z. Allybokus \& K. Avrachenkov \& J. Leguay \& L. Maggi}
\RRtitle{Partage \'equitable de ressources en temps r\'eel dans les Software Defined Networks distribu\'es}
\RRetitle{Real-Time Fair Resource Allocation in Distributed Software Defined Networks}
\titlehead{Real-Time Fair Resource Allocation in Distributed SDN}
\RRresume{
La performance des r\'eseaux informatiques est fortement li\'ee au partage \'equitable de la bande-passante entre les diff\'erents flux. 
Lorsque la taille de ces flux varie constamment dans le temps, le probl\`eme de partage des ressources est non-trivial. Afin d'aborder 
ce probl\`eme, des techniques de partage pouvant r\'eagir rapidement aux fluctuations de trafic sont d\'esirables, en particulier pour 
le contr\^ole de grands r\'eseaux avec des centaines de noeuds et des milliers de flux. Nous proposons un algorithme distribu\'e qui 
s'attaque au probl\`eme de partage de ressources \'equitable dans le contexte des architectures Software-Defined Networks (SDN) distribu\'ees.
Cet algorithme g\'en\`ere en chaque instant des solutions convergeant vers le partage \'equitable en respectant toujours l'ensemble des
contraintes, une propri\'et\'e non satisfaite par les m\'ethodes classiques de d\'ecomposition primale-duale. Gr\^ace \`a la distribution
des calculs, nous montrons que notre algorithme peut contr\^oler de grands r\'eseaux en temps r\'eel.
}
\RRabstract{
The performance of computer networks relies on how bandwidth is shared among different flows. Fair resource allocation is a challenging problem particularly when the flows evolve over time.
To address this issue, bandwidth sharing techniques that quickly react to the traffic fluctuations are of interest, especially in large scale settings with hundreds of nodes and thousands of flows. In this context, we propose a distributed algorithm that tackles the fair resource allocation problem in a distributed SDN control architecture. Our algorithm continuously generates a sequence of resource allocation solutions converging to the fair allocation while always remaining feasible, a property that standard primal-dual decomposition methods often lack. Thanks to the distribution of all computer intensive operations, we demonstrate that we can handle large instances in real-time.
}
\RRmotcle{
Software-Defined Networking (SDN), Partage \'equitable de ressources,
Maximisation de fonction d'utilit\'e de r\'eseau, $\alpha$-fairness,
Alternating Direction Method of Multipliers (ADMM),
Algorithmes distribu\'es, Distributed SDN Control Plane
}
\RRkeyword{Software-Defined Networking (SDN), Fair Resource Allocation,
Network Utility Maximization (NUM) problem, $\alpha$-fairness,
Alternating Direction Method of Multipliers (ADMM),
Distributed Algorithms, Distributed SDN Control Plane
}
\RRprojet{Neo}
\RCSophia 

\begin{document}
\makeRR   

\section{Introduction}

Software Defined Networking (SDN) technologies are radically transforming network architectures by offloading the control plane  (e.g., routing, resource allocation) to powerful remote platforms that gather and keep a local or global view of the network status in real-time and push consistent configuration updates to the network equipment. The computation power of SDN controllers fosters the development of a new generation of control plane architecture that uses compute-intensive operations. 
Initial design of SDN architectures \cite{vaughan2011openflow} had envisioned the use of one central controller. However, for obvious scalability and resiliency reasons large  networks already in production are considering a distributed SDN control plane \cite{kreutz2015software}.  Thus, a \emph{logically centralized} network control plane may consist of multiple controllers each in charge of a SDN domain of the network and operating together,  in a \emph{flat} \cite{stallings2013openflow} or \emph{hierarchical} \cite{hassas2012kandoo} architecture, to achieve global tasks.

In this paper, we study the problem of finding a global fair resource allocation when the control plane is distributed over several domain controllers. More specifically, we consider the case where the size of flows evolves over time and bandwidth allocations have to be quickly adjusted towards the novel \emph{fair} solution (in the sense of $\alpha$-fairness defined by Mo et al. in \cite{mo2000fair}). In distributed SDN architectures, controllers operate with full information in their domain and communicate (e.g., system states or optimization variables) with adjacent domain controllers or a central gathering entity. Exchanges between controllers are expensive in terms of communication delay and overhead. In this setting, a distributed algorithm may not have enough time to converge to optimum before it has to provide a feasible answer due to the scale of networks. Therefore, the main challenge is to produce very quickly good quality feasible solutions. Local mechanisms such as Auto-Bandwidth \cite{dhody2009autobandwidth} have been proposed to greedily and distributedly adjust the allocated bandwidth to support time-varying IP traffic in \emph{Multi Protocol Label Switching} (MPLS) networks. However, they do not ensure fairness and do not optimize resources globally. On the other hand, centralized algorithms have been proposed to solve the problem but fail at quickly providing good and feasible solution in a distributed setting \cite{mccormick2014real}. 

We propose a distributed algorithm that performs in real-time for the fair resource allocation problem in distributed SDN control planes. It is based on the \emph{Alternating Direction Method of Multipliers} (ADMM) \cite{boyd2011distributed} that has captured a lot of attention recently for its separability and fast convergence properties. Our algorithm, called Fast Distributed (FD)-ADMM, is designed to be fully deployed over a distributed SDN control plane and permits controllers to handle their domain simultaneously while operating together in the fashion of a general distributed consensus problem to achieve a global optimal value. We show that this algorithm can function in real-time by iteratively producing a \emph{feasible} (global) resource allocation strategy that converges to the $\alpha$-fair optimal allocation. It produces very quickly (in fact, from the first iterations) good quality feasible solutions that permit to adjust in milliseconds bandwidth for flows that evolve quickly and need immediate response, a property that standard dual decomposition methods such as the one in \cite{mccormick2014real} do not have. 

We argue that this property is crucial as the network state (e.g., flow size, flow arrival/departure, link/node congestions) may be affected by abrupt changes. Thus, we claim that it is often preferable to have a quick access to a good quality solution rather than a provably asymptotically optimal solution with poor convergence rate. Therefore, we show that our algorithm is a good candidate for real-time fair resource allocation. To this aim, we compare the performances of the algorithm to the Lagrangian dual splitting method in \cite{voice2006stability, mccormick2014real}, a standard decomposition method that violates feasibility, demonstrates poor convergence rate, hence responds slower in real-time.

Moreover, we provide an explicit and adaptive tuning for the penalty parameter in FD-ADMM so that an optimal convergence rate can be approached on any instance execution. And we show that projections can be massively parallelized link-by-link yielding a convergence rate of the algorithm that does not depend on the way the network is partitioned into domains. 

The remainder of this paper is organized as follows. Section~\ref{sec:related} surveys the related work around the real-time fair resource allocation problem. Sections~\ref{sec:problem} and \ref{sec:consensus} explicitly reformulate the fair resource allocation problem in the fashion of a general distributed consensus problem, addressed with the terminology of proximal algorithms.  Section~\ref{sec:tuning} discusses on an optimal tuning of the penalty parameter in FD-ADMM. Section~\ref{sec:execution} provides simulations that validate our approach and finally, Section~\ref{sec:conclusion} concludes the paper.

\section{Related work}
\label{sec:related}

The concept of fair resource allocation has been a central topic in networking. Particularly, \emph{max-min} fairness\footnote{A resource allocation strategy is said to be max-min fair if no route can increase its allocation while remaining feasible without penalizing another route that has a smaller or equal allocation} has been the classic resource sharing principle \cite{bertsekas1992data} and has been studied extensively. The concept of \emph{proportional fairness} and its weighted variants were introduced in \cite{kelly1998rate}. Later, a spectrum of fairness metrics including the two former ones was introduced by Mo et al. in \cite{mo2000fair} as the family of $\alpha$-fair utility functions. Some early notable work on max-min fairness includes \cite{charny1995congestion}, where the authors propose an asynchronous distributed algorithm that communicates explicitly with the sources and pays some overhead in exchange for more robustness and faster convergence. Later in \cite{skivee2004distributed}, a distributed algorithm is defined for the weighted variant of max-min fair resource allocation problem in MPLS networks, based on the well-known property that an allocation is max-min fair if and only if each \emph{Label-Switched Path} (LSP) either admits a \emph{bottleneck link} amongst its used links or meets its maximal bandwidth requirement (see Definition 4 there of a bottleneck link). The problem of Network Utility Maximization (NUM) was also addressed with standard decomposition methods that could give efficient and very simple algorithms based on gradient ascent schemes performing their update rules in parallel. In this context, Voice \cite{voice2006stability}, then McCormick et al. \cite{mccormick2014real} tackle the $\alpha$-fair resource allocation problem with a gradient descent applied to the dual of the $\alpha$-fair resource allocation problem. 

In these works, no mention is made on the potential (in fact, systematic) feasibility violation of the sequences generated by those algorithms, which we believe is a matter that deserves attention in real-time setups. Motivated by this, recently the authors of \cite{sundaresan2016iterative} provide a feasibility preserving version of Kelly's methodology in \cite{kelly1998rate}. Their algorithm introduces a slave that gives at each (master) iteration an optimal solution of a weighted proportionally fair resource allocation problem that is explicitly addressed in the case of polymatroidal and flow aggregating networks only. As a matter of fact, we contribute to this problem with an efficient \emph{real-time} version of the slave process, for any topology, preserving feasibility at each (slave) iteration. Amongst approximative approaches, one can quote the very recent work \cite{marasevic2016fast} where a multiplicative approximation for $\alpha \neq 1$ and additive approximation for $\alpha = 1$ is provably obtained in poly-logarithmic time in the problem parameters. Moreover, starting from any point, the algorithm reaches feasibility within poly-logarithmic time and remains feasible forever after. The algorithm described in our paper solves the problem optimally and reaches feasibility as from the first iteration from any starting point.

The work around ADMM is currently flourishing. The $O(\frac{1}{n})$ best known convergence rate of ADMM \cite{he2012douglas} failed to explain its empirical fast convergence until very recently, for instance in \cite{Deng2016}, where global linear convergence rates are established in four scenarios of the strongly convex case. ADMM is also well-known for its performance that highly depends on the parameter tuning, namely, the penalty parameter in the augmented Lagrangian formulation (see Section \ref{sec:problem} below). An effective use of this class of algorithms cannot be decoupled from an accurate parameter tuning, as convergence can be extremely slow otherwise. Thus, in the same paper  \cite{Deng2016}, the authors provide a linear convergence proof that yields a convergence rate in a closed form that can be optimized with respect to the objectives parameters. Therefore, thanks to these works, an optimal tuning of ADMM for $\alpha$-fair resource allocation is now available. Several papers use the distributivity of ADMM to design efficient algorithms solving consensus problems in e.g. model predictive control and congestion control, without however addressing this fundamental detail. In the simulations of \cite{mota2012distributed} for instance, the authors try several choices of the penalty parameter and plot the best result found for each point.

\noindent{\textbf{Our contribution:}} In this paper, we reformulate the $\alpha$-fair resource allocation problem and we design a distributed algorithm based on ADMM, called FD-ADMM. We show that this algorithm outputs at each iteration a feasible resource allocation strategy that converges to the unique optimum of the problem. We also provide an adaptive strategy to correctly tune the FD-ADMM penalty parameter and we show that projections can be  massively parallelized on a link-by-link basis.
Finally, we show how our algorithm outperforms the dual methods mentioned above in terms of feasibility preservation and responsiveness in dynamic scenarios.


\section{Fair resource allocation problem}
\label{sec:problem}
In this section, we reformulate the $\alpha$-fair resource allocation problem as a convex optimization problem. Then, we start off with our algorithm design by presenting C-ADMM, an algorithm that solves our problem in a centralized fashion and that will be helpful to design our distributed algorithm.

\subsection{Problem reformulation}
\label{subsec:problemformulation}
Let $R$ be a set of connection requests over a network with a set $J$ of capacitated links. Each link $j\in J$ has a total capacity of $C_j \in \RR_+$. Each request $r$ is represented by a route containing a subset of $J$ that, without any confusion, we still denote as $r$. With some abuse of notation, we write $j\in r$ or $r\in j$ to say that link $j$ belongs to the route $r$, or route $r$ goes through link $j$, respectively. Given the set of requests and their corresponding utility function $f_r$, the network allocates bandwidth to all the requests in order to maximize the overall utility $f = \bigoplus_r f_r$, while satisfying feasibility, i.e., the link capacity constraints. Denote by $x_r$ the capacity allocated to route $r$, and let $x = (x_r)_{r\in R}$. Then, we have the classic capacity constraint in matricial form:
\begin{equation}
\label{eq:capconU}
Ax \leq C
\end{equation}
where $A = (a_{jr})_{jr}$ is the link-route incidence binary matrix:

$$a_{jr} = \left\{ \begin{array}{ll}
1  & \mbox{if } j\in r\\
0 & \mbox{otherwise}.
\end{array}\right.$$

\noindent Our aim is to compute an $\alpha$-fair capacity allocation $x$:
\begin{equation}
\label{eq:fairness}
\max_{x\geq 0, Ax\leq C}
f^\alpha(x) 
\end{equation}
where the $\alpha$-fair utility function $f^\alpha$ is defined according to the Mo and Walrand's classic characterization in \cite{mo2000fair}, that we report below.

\begin{definition}[$(w,\alpha)$-fairness, \cite{mo2000fair}]
\label{defalphafair}
Let $F\subset \RR_+^n$ be a non-empty feasible set not reduced to \{0\}.
Let $w \in \RR_+^n$ and $x^*\in F$. We say that $x^*$ is $(w,\alpha)$-\emph{fair} (or simply $\alpha$-fair when there is no confusion on $w$) if the following holds:
\[\forall r\in [1,n],\quad x^*_r >0\quad\mbox{and}\quad\forall x\in F 
,\quad \sum_{r=1}^n w_r\frac{x_r - x^*_r}{x^{*\alpha}_r} \leq 0. \]
Equivalently, $x^*$ is $(w,\alpha)$-fair if, and only if $x^*$ maximizes the $\alpha$-fair utility function $f^\alpha$ defined over $F-\{0\}$:
\[f^\alpha(x) = \sum_{r = 1}^n f^\alpha_r(x_r),
\]
where
$
f_r^\alpha(x_r) =
\left\{ \begin{array}{ll}
w_r \frac{x_r^{1-\alpha}}{1-\alpha}, & \alpha \neq 1,\\
w_r \log(x_r), & \alpha=1.
\end{array}\right.
$
\end{definition}

The success of $\alpha$-fairness is due to its generality: in fact, for $\alpha=0,1,2,\infty$ it is equivalent to max-throughput, proportional fairness, min-delay, and max-min fairness, respectively. We observe that the $\alpha$-fair utility functions are non-decreasing, strictly concave, non-identically equal to $-\infty$, and upper semi-continuous. It is well-known that under these conditions, the function $f^\alpha$ admits a unique maximizer over any convex closed non-empty set.

From now on, we adopt the convex optimization terminology. Define for each $r\in R$ the convex cost function $g_r:~x_r~\mapsto~ g_r(x_r) := -f_r(x_r)$. Then,  $g :=  \bigoplus_r g_r = - f^\alpha $ is a convex closed proper\footnote{\emph{closed} stands for lower semi-continuous and \emph{proper} means non-identically equal to $\infty$} function over $\RR_+^{|R|}$. 
%
Let us introduce $\iota$ as the indicator function of the convex closed set  $\{Ax\leq C, x\geq 0\}_x$:
$$\iota(x) =
\left\{ \begin{array}{ll}
0 & \mbox{if }Ax\leq C\\
\infty, & \mbox{otherwise}.
\end{array}\right.$$ 
Then our $\alpha$-fair problem can equivalently be formulated as the following convex program:

\begin{equation}
\min_{x,z} \sum_{r\in R}g_r(x_r) + \iota(z),
\label{ConvObj}
\end{equation}
\begin{equation}
\mbox{s.t.}\quad x-z = 0.
\label{ConvConst}
\end{equation}

\subsection{ADMM as an augmented Lagrangian splitting}
\label{augmentedlag}

Let us begin by recalling to the reader the basic principles of the \emph{Alternating Direction Method of Multipliers} (ADMM), applied to our $\alpha$-fair problem.
To this aim, the augmented Lagrangian with penalty $\lambda^{-1}>0$ for problem~(\ref{ConvObj}-\ref{ConvConst}) writes\footnote{$a^\T b$ is the Euclidean product of $a$ and $b$ and $||\cdot||$ the Euclidean norm.}
\begin{equation}
L_{\lambda^{-1}}(x,z,u) = g(x) +\iota(z) + u^\T(x-z) + \frac{1}{2\lambda} ||x-z||^2
\end{equation}
where $u$ is the vector of Lagrange multipliers.
The method of multipliers consists in the following update rules, where the superscript $k$ denotes an iteration count:

\begin{equation}
(x^{k+1},z^{k+1}) = \argmin{x,z} L_{\lambda^{-1}}(x,z,u^k) \tag{M1}
\label{M1}
\end{equation}
\begin{equation}
u^{k+1} = u^k + \frac{1}{\lambda}(x^{k+1} - z^{k+1}). \tag{M2}
\end{equation}

The main idea in alternating directions is in fact to decouple the variables $(x,z)$ in the optimization stage~\ref{M1}: instead of a global optimization over $(x,z)$, we only optimize $L_{\lambda^{-1}}$ with respect to the variable $x$, then, given the new update of $x$, we optimize $L_{\lambda^{-1}}$ with respect to $z$. Before stating the corresponding update rules of ADMM, let us first remind the following Fact.

\begin{fact}[\cite{boyd2011distributed}]
\label{fact:def}
Let $h: \RR^n \to \bar{\RR} = \RR\cup \{\infty\}$ be a closed proper convex function. The set $\mathrm{dom}(h)$ denotes the domain of $h$, that is the set upon which $h$ takes real values.
Assume $\mathrm{dom}(h) \neq \emptyset$. Then, the following facts hold:
\begin{enumerate}[\itshape (i)]
\item For $u\in \RR^n$, $\lambda \in \RR^*_+$, the minimization problem $$u^*_\lambda= \argmin{x}\left\{h(x) + \frac{1}{2\lambda} ||u-x||^2\right\}$$ admits a unique solution. The ($\lambda$-scaled) \emph{proximal operator} of $h$ is the well-defined map $\prox{\lambda h}: u \to u^*_\lambda$.

\item Assume that $h$ takes the form $h(x,y) = h_1(x)+ h_2(y)$, for $(x,y)\in \RR^p\times \RR^{n-p}$ (write $h = h_1 \bigoplus h_2$) where $h_1, h_2$ are both closed, proper and convex. Then, for $(u,v) \in \RR^p\times\RR^{n-p}$,  $\prox{\lambda h}(u,v) =\left( \prox{\lambda h_1} (u), \prox{\lambda h_2}(v)\right)$.

\item Assume that $h$ is the indicator function of a closed convex non-empty set $F$.
Then $P_F := \prox{\lambda h}$ is the Euclidean projection onto $F$.
\end{enumerate}
\end{fact}
The definition of a proximal operator being set, a straightforward calculus shows that we have:
\begin{equation}\nonumber
\forall x,u\quad \argmin{z} L_{\lambda^{-1}}(x,z,u) = \prox{\lambda\iota}(x + \lambda u)
\end{equation}
\begin{equation}\nonumber
\forall z,u \quad \argmin{x} L_{\lambda^{-1}}(x,z,u) = \prox{\lambda g}(x - \lambda u).
\end{equation}
ADMM can thus be expressed in the proximal ($\lambda$-scaled) form, which we refer to as \emph{Centralized} ADMM (C-ADMM).

\begin{algorithm}[H]
\caption{Centralized ADMM (C-ADMM)}
\label{alg:CADMM}
\begin{algorithmic}[1]
\Require{Initial values  $z$, $v$}
\While{a suitable termination condition is not met}
	\State $x \gets \prox{\lambda g}(z - v)$ \label{CADMMop1}
	 \State $z \gets P(x + v)$ \label{CADMMop2}
	\State $v \gets v +x - z$
\EndWhile
\end{algorithmic}
\end{algorithm}
In Algorithm~\ref{alg:CADMM}, $P = \prox{\lambda \iota}$ is the projection on $\{Ax\leq C, x\geq 0\}_x$, and $v = \lambda u$ the  $\lambda$-scaled dual variable. Now, the first step of Algorithm~\ref{alg:CADMM} (line~\ref{CADMMop1}) can be separated thanks to the separability property of the objective function, see Fact~\ref{fact:def}. In fact, $g$ is fully separable, as $g(x) = \sum g_r(x_r)$. Thus, the proximal update of line~\ref{CADMMop1} takes the trivially parallelized form:

\begin{equation}
\forall r\quad x_r^{k+1} =  \prox{\lambda g_r}(z_r^k - u_r^k)\label{CADMMx}
\end{equation}
such that each local variable $x_r$ can be computed separately. 

Through expression~\eqref{CADMMx}, we are thus able to provide an efficient update rule for $x$, provided that the separate proximal computations are inexpensive. However, two main issues arise. 

\noindent \textbf{Main issues with C-ADMM:} \textit{a)} First, an update of the variable $z$ in line~\ref{CADMMop2} of Algorithm~\ref{alg:CADMM} requires full knowledge of the projection mapping, which in turn requires full information on the capacity set of the network. Thus, this global update rule represents an important limiting factor to the design of a fully distributed algorithm, which is our main design interest here to follow the distribution of SDN control planes. 

\textit{b)} Moreover, although the convergence of C-ADMM may only require some tens of iterations (see Section~\ref{sec:execution} for further details), it may be slow in terms of computation time due the successive application of a projection algorithm that would not scale with respect to the problem size. This also gives rise to a double loop algorithm where each iteration requires the convergence of an inner process that can be time-consuming. Indeed, computing the projection of a generic point onto a closed convex non-empty polyhedron is in general non-trivial. Hence, for general polyhedra, one has to operate alternate projections, summon quadratic programming solvers or use  iterative algorithms such as the one in \cite{iusem1987simultaneous}. 

We address issues \textit{a,b)} in the next section, where we  propose FD-ADMM, a distributed version of C-ADMM. 


\section{The general consensus form of ADMM: an efficient distributed algorithm design}
\label{sec:consensus}

In this section, we show how to alleviate the cost of the global projection sub-routine in C-ADMM (line~\ref{CADMMop2}) by decomposing the formulation with respect to the network links of each SDN domain in the fashion of a consensus problem, and present FD-ADMM.  As stated at end of Section~\ref{sec:problem}, the global knowledge of the topology and the computational effort required by the projection step (line~\ref{CADMMop2}) of C-ADMM are not affordable in the distributed SDN control plane. Thus, the decomposition permits to respect the locality of the different domain controllers that now handle the projections link by link efficiently and in parallel. The decomposition into domains can be orchestrated by the SDN architect without any constraint. Unavoidably though, domains will need to exchange information as routes may traverse different domains. 

\subsection{Preliminaries}
\label{prel}
We organize the network into several domains $J_p, p = 1\ldots P$ such that $(J_p)_p$ forms a partition of the set of links $J$. Let $R_p$ be the set of routes traversing the domain $J_p$ via some link $j\in J_p$. More formally, $R_p = \{r\in R:  \exists \,j\in J_p \mbox{ s.t. }  j\in r\}$. Hence, $(R_p)_p$ forms a covering of $R$.
Let $\iota_j$ denote the indicator function for link $j \in J_p$, i.e.,

\begin{equation}
\iota_j(x) =
\left\{ \begin{array}{ll}
0 & \mbox{if} \ \sum_{r\in j} x_r \leq C_j \\
\infty & \mbox{otherwise}. 
\end{array}\right.
\end{equation}

%
Also, let us define $S_j := \dom{\iota_j}$.
\noindent Thus, for each $j\in J$, $S_j$ 
 is the (convex, closed) capacity set of the link $j$. Finally, for $j\in J$ and $z \in \RR^{R}$,  $\mbox{\textsc{Projection}}(j,z)$ denotes the Euclidean projection of $z$ onto $S_j$.



 
\subsection{Consensus form}

We can now reformulate our objective to a fully separable form. For ease of notation, the variable $x$ will be written $z_0$ and we define $R_0 = R$. We also define an additional variable, $\tilde{z}_r$, that will represent the consensus value of $z_{0r}$ found for each route $r$ over all the domains handling the route $r$.
We write $I_r = \{q\in [0,P] \quad r\in R_q\}$ to design the set of domain indices (including index 0) which $r$ belongs to. 
In the same fashion as in Section~\ref{augmentedlag}, we plug the feasibility constraints into the objective. Each constraint being now handled separately, we can formulate Problem~\eqref{ConvObj}, \eqref{ConvConst} as follows:

\begin{equation}
 \min \sum_{r\in R} g_r(z_{0r}) + \sum_{j\in J} \iota_j(z_0).
\end{equation} 
 
 In order to obtain a separable objective and fully benefit from the separability property in Fact~\ref{fact:def}, we artificially create a copy of the variable $z_0$ for each link $j$. This variable will be handled by the unique domain $J_p$ containing $j$. For each $j$, let $z_j \in \RR^{|R|}$ be the copy of $z_0$ for link $j$.

%
Creating a complete copy of all the variables for each domain is, nevertheless, of no use. Each domain indeed only needs information and manipulation over the only variables associated with the routes that they handle completely (the route is included in the domain's links)  or partially (the route meets other domains). Now, $\iota_j$ actually depends only on the sub-variable $(z_j)_{R_j}\overset{\mathrm{def}}{=}(z_{jr})_{r\in R_j} $. We erase all the information that is irrelevant to region $J_p$: $z_j \in \RR^{R_j}$. We can thus write the objective as follows:
 
 \begin{equation}
 \min G(z) = \sum_{r\in R} g_r(z_{0r}) + \sum_{j\in J} \iota_j((z_j)_{R_j}).
 \end{equation}
 
To sum up, \emph{we have artificially separated the objective function by creating a minimal number of copies of the primal variable $z_0$ in order to fully distribute the problem}. Now, instead of a global resource allocation variable, several copies of the variable account for how its value is perceived by each link of each domain. 
To enforce an intra- (local) and inter- (global) domain consistent value of the appropriate allocation, consensus constraints are added to the problem. This new formulation can be interpreted as a multi-agent consensus problem formulation where route $r$ has cost $g_r$, and link $j$ has cost $\iota_j$. As we separated the global objective on purpose, the separability property of the proximal operator thus gives the following:
$$\prox{\lambda G}(u) = \left( (\prox{\lambda g_r}(u_{0r}))_r,  (\mbox{\sc{Projection}}(j,u_j))_j \right).$$
These considerations permit next to write our final distributed consensus model where each agent only has access to local information.
  
\subsection{Fast Distributed ADMM}
  
We can finally distribute ADMM by putting into practice the tricks described in the previous section. Then, the general consensus form of the problem can be expressed as follows.
 
  \begin{equation}
 \min \sum_{r\in R} g_r(z_{0r}) + \sum_{j\in J} \iota_j(z_j)
 \end{equation}
 \begin{equation}
 z_{jr} = z_{lr} \quad \forall r\in R_j\cap R_l \quad \forall j,l \in \{0\}\cup J
 \label{general consensus}
 \end{equation}
where $z_j = (z_{jr})_{r\in R_j} \in \RR_+^{|R_j|}$. By applying ADMM to this formulation and using again Fact~\ref{fact:def} we obtain, after some simplification, Algorithm~\ref{FDADMM} (Fast Distributed (FD)-ADMM). \\
To update the consensus variables $\tilde{z}_r$, we exploit the fact that the Euclidean projection of a point $y\in \RR^n$ onto the diagonal is simply its average $\frac{1}{n} \sum y_i \mathbf{1}$. Hence, if $I$ denotes the indicator function of the feasible set~\eqref{general consensus}, we have: $$\forall r\in R \quad \prox{\lambda I} (u)_r = \frac{1}{|J_r|+1}\left( \sum_{l \in r} u_{lr} + u_{0r}\right).$$  This yields the simple update rules at lines~\ref{alg2:comm} and \ref{send}\footnote{These updates rules are also simplified using the straightforward fact that the sum $\sum_{l\in r}  u_{lr}$ is constant. It can thus be fixed to $0$ by initialization.}.\\ Notably, even in the distributed case, each domain $p$ can compute at each iteration a \emph{globally} feasible allocation $z_{*r}$ for each of the routes $r \in R_p$ (see Proposition~\ref{claim:feasible}).
 

 \begin{algorithm}[t]
 \caption{Fast Distributed ADMM (FD-ADMM)}\label{FDADMM}
 \begin{algorithmic}[1]
 \Procedure{$\mbox{\textbf{of} \textsc{ Domain} $p$}$}{}
 \Require{Reciprocal penalty parameter $\lambda$, $(g_r)_{r\in R_p}$}
 \State \mbox{\sc Receive} $z_{qr}, z_{*qr}\quad \forall q\in I_r\quad \forall r \in R_p$\label{receive} 
 \State $ \mbox{\bf\sc Enforce}\quad z_{*r} = \min_{q\in I_r}  z_{*qr}\quad \forall r\in R_p$
 \State $ \tilde{z}_{r} \gets \frac{1}{|J_r|+1} \left(\sum_{q\in I_r}
 z_{qr}+  z_{0r}\right) \quad \forall r\in R_p$\label{alg2:comm}
 \For{$j\in J_p \cup \{0\}$}
 	\State $ u_{jr} \gets u_{jr} +z_{jr} -\tilde{z}_{r}\quad \forall r\in R_j$
 	\State $z_j \gets \mbox{\sc Projection}(j,\tilde{z} - u_j)$ \label{projectionstep}
\EndFor
\State $ z_{0r} \gets \prox{\lambda g_r}(\tilde{z}_{r} - u_{0r})\quad \forall r \in R_p $\label{updatecentral}

\State \mbox{\sc Send} $z_{pr} = 
\sum_{j\in J_r \cap J_p} z_{jr}$ and $ z_{*pr} = \min_{j\in J_p} z_{jr}$ to domains $q \in I_r \quad \forall r\in R_p \label{send}$

\EndProcedure
 \end{algorithmic}
 \end{algorithm}

\noindent \textbf{Communication among domain controllers:} In FD-ADMM, \emph{only domains that do share a route together have to communicate}. The communication procedures among the domain controllers are described at lines~\ref{receive} and \ref{send}. In these steps, the domains gather from and broadcast to adjacent domains the sole information related to routes that they share in common. In particular, domains are blind to routes that do not traverse them, and can keep their internal routes secret from others. In details, after each iteration of the algorithm, each domain $J_p$ receives the minimal information from other domains such that $J_p$ is still able to compute a local value $z_{pr}$ and a locally feasible value $z_{*pr}$. Next, $J_p$ send them back to neighboring domains $I_r$ that $r$ traverses.

\emph{Communication overhead:} In terms of overhead, we can easily evaluate the number of floats transmitted between each domain at each iteration. At each communication, domain $J_p$ must transmit $z_{pr}$ and $z_{*pr}$ for each $r\in R_p$ to each other domain that $r$ traverses. The variable $z_0$ does not need to be centralized or transmitted between controllers. Each domain controller may actually have a copy $z_0$ and perform the (low-cost) computation of their update rule (see line~\ref{updatecentral} in Algorithm~\ref{FDADMM}) locally. Hence, domain $p$ transmits in total $2 \sum_{q \neq p} |R_p\cap R_q|$ floats to the set of its peers. As a comparison, in a distributed implementation of the algorithm given in \cite{mccormick2014real} and stated in Section~\ref{sec:execution}, each domain $p$ transmits in total $\sum_{q\neq p} |\{j\in J_p, \exists r\in R_q \mbox{ s.t. }  j\in r\}|$ floats to the set of its peers, which is bounded by $(P-1)|J_p|$ as $|R|$ grows. 

\noindent \textbf{Feasibility preservation:} A potential drawback of the distributed approach is the potential feasibility violation by the iterate $\tilde{z}^k$. However, we have the following positive result.
\begin{claim}
\label{claim:feasible}
FD-ADMM provides a sequence of feasible points that converges to the optimum.
\end{claim}
\begin{proof}
Consider the iteration number $k$ and drop the superscript $k$ for lightness. For any link $j$, we have by line~\ref{projectionstep} of Algorithm~\ref{FDADMM} that $z_j$ is feasible in link $j$. That is, $\sum_{r\in j} z_{jr} \leq C_j$. Define $z_{*r} = \min_{j\in J_r} z_{jr}$. Then, for each link $j$:
\begin{equation}
\sum_{r\in j} z_{*r} \leq \sum_{r\in j} z_{jr} \leq C_j.
\end{equation}
Thus, no capacity is violated by the allocation $z_{*r}$. At the optimum, the consensus is reached. Thus $(z^k_{*r})_k$ is a feasible sequence that converges to the optimum.
\end{proof}
\noindent The number $z_{*r}$ introduced in Proposition~\ref{claim:feasible} above in fact corresponds to the introduced variable of the same name FD-ADMM. Thus, in a certain way, for sufficiently loaded and communicating domains (i.e. the $|R_p\cap R_q|$ are large enough) we sacrifice some overhead (counted on a per iteration basis) compared to standard dual methods, but in exchange for anytime feasibility, a major feature that dual methods do not generically provide.

\section{Implementation and algorithm tuning}
\label{sec:tuning}
In this section, we discuss two major points in the design of FD-ADMM. First, we precise and justify the choice of the procedure \textsc{Projection}, in line 7 of FD-ADMM Algorithm \ref{FDADMM}. Next, we derive an explicit adaptive update of the reciprocal penalty parameter $\lambda$ that permits to accelerate the convergence of FD-ADMM on any instance.

\subsection{Projection procedure: A discussion}

In Section \ref{sec:consensus}, we advocated a link-wise separation of the formulation because it is non-trivial to project an arbitrary point onto an arbitrary closed convex polyhedron. 
However, the projection onto the sets $S_j$ (see Section~\ref{prel}) can be done with an exact method with a complexity dominated by the one of sorting a list of the size of its dimension. In average, sorting a list of length $q$ is done in $O(q\log q)$. Hence, by operating instead a link-by-link projection, the controllers save a huge amount of time by providing an (generically infeasible) approximate projection point $z_{pr}$ and deriving a locally feasible allocation $z_{*pr}$ (see Algorithm~\ref{FDADMM} line~\ref{send}). Although the quality of the global iterate $z_*$ may be altered by further distribution of the projection, the point is quickly generated. Paradoxically enough, FD-ADMM therefore fully adapts to any network distribution into domains \emph{because} it functions by link, regardless of the network partition into domains. The algorithm we use for \textsc{Projection} in FD-ADMM is presented for instance in \cite{chen2011projection} in which the authors also give a correctness proof and performance demonstration. It permits to provide an efficient update for each domain $J_p$.

%
%

\subsection{Estimating the optimal parameter $\lambda$ }

\label{sec:parameter}

It is well-known that the reciprocal penalty parameter $\lambda$ highly conditions the convergence speed of ADMM. An inaccurate tuning can indeed lead to a very slow convergence. For appropriate problems, it is possible to use a result proven in \cite{Deng2016} to compute an optimal reciprocal penalty parameter, that we here report. It will help us tune FD-ADMM to optimize its convergence performance\footnote{We recall that a differentiable function $f: \RR^n \to \bar{\RR}$ is \emph{strongly convex} with modulus $\sigma$ if $(\nabla f(x) - \nabla f(y))^\T  (x-y)\geq \sigma ||x-y||^2 ,\quad \forall x,y \in \mathrm{dom}(f)$. Moreover, $f$ is \emph{Lipschitz} with modulus $L$ if  $|f(x) - f(y)|\leq L|x-y|, \quad \forall x,y \in \dom{f}$.}. 



\begin{theorem}[\cite{Deng2016}] \label{theo:deng}
Assume that the following problem:
\begin{align}
& \min F(x) + G(z) \tag{M} \label{mat} \\
& \mbox{s.t. } Mx-Pz = 0 \notag
\end{align} 
has a saddle point, and both objective functions are convex. Assume that $M$ has full row rank, and that  $F$ is $\sigma$-strongly convex and has a $L$-Lipschitz gradient. Then, the sequence of iterates (primal and dual concatenated) of ADMM converges linearly with rate $(1+\delta)^{-1}$, where\footnote{$\lambda_{\min}$ is the smallest eigenvalue of a positive matrix, and $||M||$ is the operator norm} $$\delta = 2\left(\frac{ ||M||^2}{\lambda\sigma} +\frac{L \lambda}{\lambda_{\min}(M^\T  M)}\right)^{-1} $$ and $\frac{1}{\lambda}$ is the penalty parameter in the augmented Lagrangian form (see Section \ref{augmentedlag}).
\end{theorem}

The following result directly follows.
\begin{corollary} \label{cor:1}
The optimal reciprocal penalty parameter is $$\lambda_* =  \displaystyle\sqrt{\frac{||M||^2 \lambda_{\min}(M^\T M)}{\sigma L}}.$$ 
\end{corollary}

In order to be able to apply Corollary \ref{cor:1}, we still need to express the coefficients of interest $\sigma,L_d$. 
\begin{fact}
\label{moduli}

The function $g=\bigoplus_r{ g_r}$ is $\sigma$-strongly convex and has $L_d$-Lipschitz gradient with: $$\sigma = \alpha \min_r \frac{w_r}{B_r^{\alpha +1}}, \quad L_d = \alpha\max_r{\frac{w_r}{d_r^{\alpha+1}}}$$ on any compact subset of $\mathrm{dom}(g)$ of the form $K_d = \{x \geq d, Ax\leq C\}$, for $d\gg0$, and $B_r = \min_{j\in r} C_j$.
\end{fact}

\begin{proof}

Consider the case $g = g_r$. We start with the calculus of $\sigma$. We write $g(x) = -w \frac{x^{1-\alpha}}{1-\alpha}$, when $\alpha \neq 1$, and $g(x) = -w \log(x)$ when $\alpha = 1$. Suppose $\alpha > 1$. For $y<x<t \in \mathrm{dom}(g)$, we have:
\begin{eqnarray*}
(\nabla g(x) - \nabla g(y)) (x-y) & =& w(\frac{1}{y^\alpha} - \frac{1}{x^\alpha})(x-y)\\
				      & = &\frac{w}{x^\alpha y^\alpha}(x^\alpha - y^\alpha)(x-y)\\
				      & = & \frac{w}{x^\alpha y^\alpha}\alpha c^{\alpha - 1}(x-y)^2), \quad (y<c<x)\\
				      & = &\frac{\alpha w}{x^\alpha y}\left(\frac{ c}{y}\right)^{\alpha -1}|x-y|^2\\
				      & \geq & \frac{\alpha w}{x^{\alpha}y} |x-y|^2\\
				      & \geq & \frac{\alpha w}{t^{\alpha+1}} |x-y|^2,
\end{eqnarray*}
where the third equality is just an application of the mean value theorem. The case $\alpha <1$ is handled likewise by integrating $x^{\alpha-1}$ into the parenthesis instead of $y^{\alpha-1}$ in the fourth line. The case $\alpha = 1$ is straightforward. By plugging in $x,y$ an appropriate sequence, say, respectively, $t-\frac{1}{n}$ and $t-\frac{2}{n}$, one can see that this bound is tight.

As for the Lipschitz factor, similarly, take $s<y<x$, we have:
\begin{eqnarray*}
|\nabla g(x) - \nabla g(y)|  & =& w|\frac{1}{y^\alpha} - \frac{1}{x^\alpha}|\\
				      & = &\frac{w}{x^\alpha y^\alpha}|x^\alpha - y^\alpha|\\
				      & = & \frac{w}{x^\alpha y^\alpha}\alpha c^{\alpha - 1}|x-y|\\
				      &  & \qquad (y<c<x)\\
				      & \leq & \frac{\alpha w}{s^{\alpha+1}}|x-y|,
\end{eqnarray*}
where the last line is obtained in the same fashion as for the calculus of $\sigma$, for each case $\alpha <1$, $\alpha > 1$. The case $\alpha = 1$ is straightforward.
Now, consider $g = \bigoplus w_r g_r$. For $x,y\in \dom{g}$,
\begin{eqnarray*}
						& & \left(\nabla g(x) - \nabla g(y)\right)^\T  \left(x-y\right)\\
						 & =& \sum_r (\nabla g_r(x_r) - \nabla g_r(y_r))(x_r - y_r)\\
						  & \geq & \alpha \sum_r \frac{w_r}{B_r^{\alpha +1}}|x_r - y_r|^2\\
						  & \geq & \alpha \min_r\frac{w_r}{B_r^{\alpha +1}} \sum_r |x_r - y_r|^2\\
						  & = & \alpha \min_r\frac{w_r}{B_r^{\alpha +1}} ||x-y||^2.
\end{eqnarray*}
The derivation of $L_d$ is similar.\qedhere\end{proof}

Unfortunately, Corollary \ref{cor:1} cannot be directly applied to our general consensus formulation. Indeed, its matricial formulation does not provide a full-row rank matrix $M$. The problem which the Theorem \ref{theo:deng} applies to is actually the original, centralized one in (\ref{ConvObj}-\ref{ConvConst}). Therefore, \emph{we will derive a reciprocal penalty parameter selection for the centralized problem, and use it as a tool to estimate a satisfactory parameter for FD-ADMM}.
 
However, the last difficulty we encounter in choosing the optimal reciprocal penalty parameter is to correctly evaluate the Lipschitz modulus. Unfortunately, $\nabla g$ is \emph{not} Lipschitz on the feasibility set, because of the singularity of each $g_r$ at $0$. In order to circumvent this problem, we introduce the classic concept of \emph{disagreement point} $d$, according to \emph{bargaining theory} terminology. 
A disagreement point $d$ represents the minimal values for an allocation of each route. This allows to reduce the feasibility set to a compact subset of the form $K_d, d\gg 0$, on which $\nabla g$ is now Lipschitz. The disagreement point can be naturally defined as the feasible point $z_*$ at the first iteration. Generically, there is no \emph{a priori} guarantee that the set $K_{z^*}$ contains the optimum, but, we remark that at least in the first iterations, the use of $z_*$ provides a good approximation of the best reciprocal penalty parameter. 
The analytical evaluation of this phenomenon goes beyond the scope of this paper and we keep it for future work. 

Thus, finally, we update $\lambda$ in an adaptive fashion in the beginning of the algorithm with the help of those points. 
We found empirically that operating such update only at the initial steps of FD-ADMM and then fixing $\lambda$ for the rest of the execution provides a good performance in terms of convergence speed. In the next section, we describe this typical phenomenon in Figure~\ref{fig:stepsize}. In all our simulations, we use the simple following update scheme to estimate the optimal penalty parameter at each execution of the algorithm. 

\begin{scheme}[Reciprocal Penalty Adaptation]
\label{alg:penalty_adaptation}
Set threshold $\tau$. At all iterations below $\tau$, denote by $p$ the last output of a feasible point. Then, choose the new reciprocal penalty parameter as:
$$ \lambda_* = \frac{1}{\alpha}\left(\min_{r\in R} \frac{w_r}{B_r^{\alpha+1}} \max_{r\in R}\frac{w_r}{p_r^{\alpha +1}}\right)^{-\frac{1}{2}}.$$
After $\tau$ iterations, do not update $\lambda$.
\end{scheme}
In our numerical evaluations we will set $\tau=30$. Thus, FD-ADMM is now fully tuned and we are ready to demonstrate its performance in the next section, in terms of convergence speed in real-time scenarios.


 
\section{Performance analysis}  \label{sec:execution}

    \begin{figure*}[t!]
    \hspace{-5mm}
   \centering
\includegraphics[width = .95\textwidth]{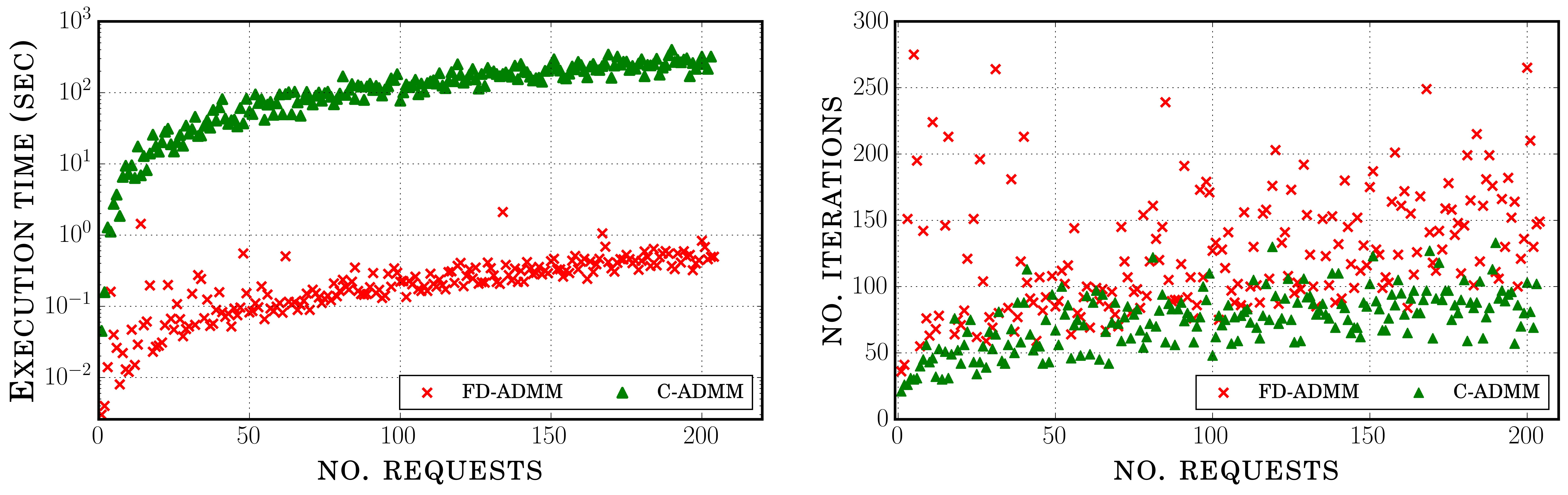}
\caption{C-ADMM against FD-ADMM: execution time and iteration count.}
 \label{fig:itercvsfd}
 \end{figure*}

We now evaluate numerically FD-ADMM in terms of its convergence properties. More specifically, in Section \ref{algodesign} we compare the performance of FD-ADMM and C-ADMM in offline scenarios where the optimum is desired. In Section  \ref{responsiveness} we evaluate FD-ADMM in real-time scenarios, where good and feasible solutions are needed on-the-fly as weights $w_r$ vary over time. In order to benchmark the transient properties of FD-ADMM we use the standard Lagrangian dual decomposition approach (LAGR) for single-path routing in \cite{voice2006stability, mccormick2014real,palomar2006tutorial}, that we recall in Algorithm~\ref{Kelly}. We here assume that domain controllers operate in synchronous mode. In this case, the decomposition into domains has no impact on FD-ADMM performance, as projection is on a link-basis. 
All simulations are made for the proportional fairness objective functions ($\alpha = 1$). We used the proximal operation formulas found in \cite{boyd2011distributed}.  
The algorithms under investigation were evaluated using BT's 21 CN network topology\footnote{We would like to thank the authors of \cite{mccormick2014real} for their willingness to share the BT 21 CN topology dataset.}, containing 106 nodes and 474 links. The requests were generated by computing the shortest path between randomly chosen sources and destinations.
\vspace{-1mm}
\begin{algorithm}[b]
	\caption{Lagrangian-based gradient descent (LAGR)}
	\label{Kelly}
	\begin{algorithmic}
		\Require{Initial positive values $u_j$}
		\While{a suitable termination condition is not met}
		\State $x_r \gets \argmax{x \geq 0} \{ f_r(x) - x \sum_{j: j \in r} u_j)\} \quad \forall r$
		\State $u_j = u_j - \frac{u_j}{2 C_j} (C_j - \sum_{r: j \in r} x_r) \quad \forall j$
		\EndWhile
	\end{algorithmic}
\end{algorithm}


 \begin{figure}[h!]
 \centering
\includegraphics[width = .47\textwidth]{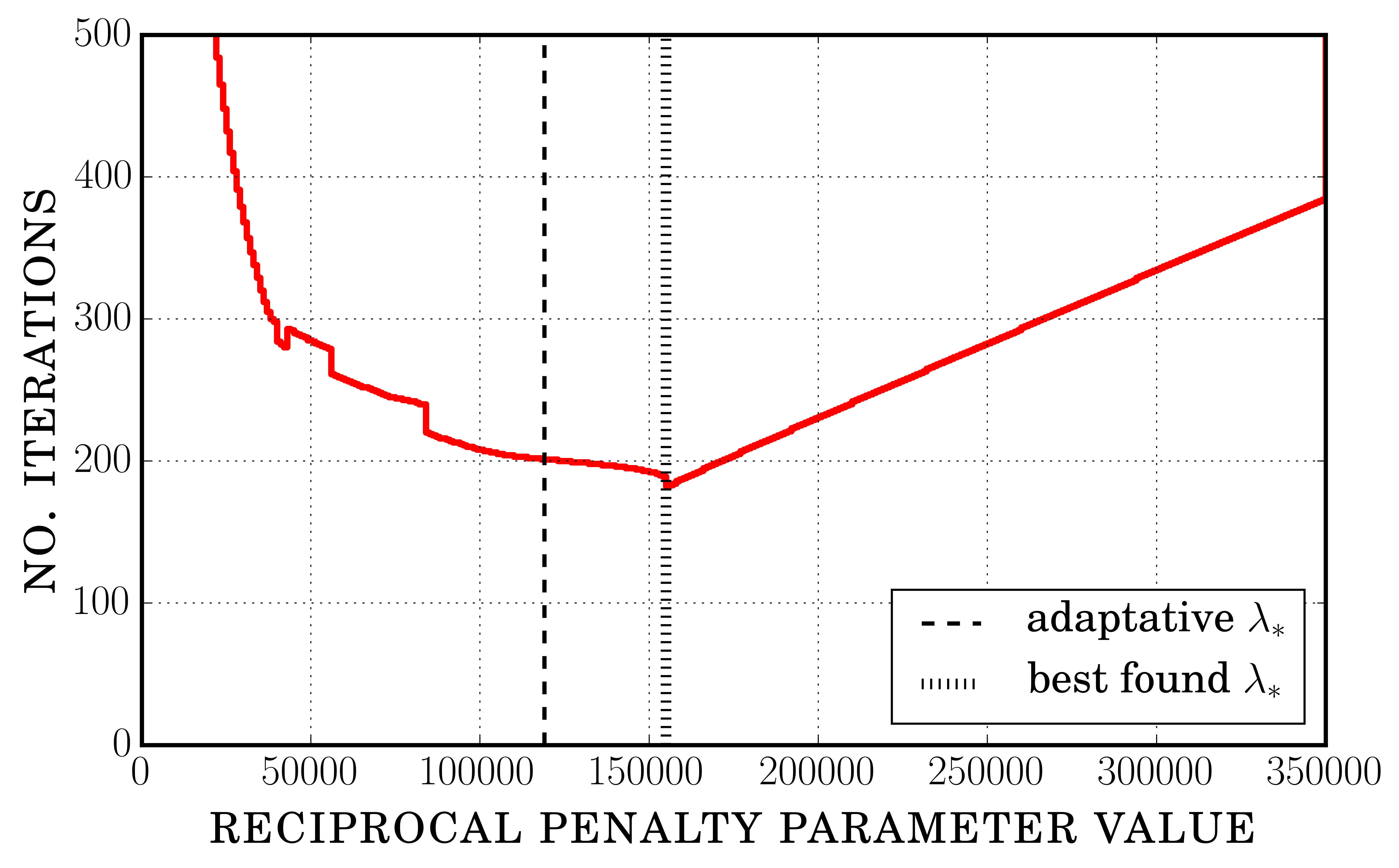}
\hspace{1mm}
\caption{Convergence rate of FD-ADMM vs. reciprocal penalty parameter. The adaptive approximation demonstrates sufficient accuracy.}
 \label{fig:stepsize}
 \end{figure}
\begin{figure}[t]
\centering
\vspace{-4.5mm}
\hspace{-5mm}
\includegraphics[width = .52\textwidth]{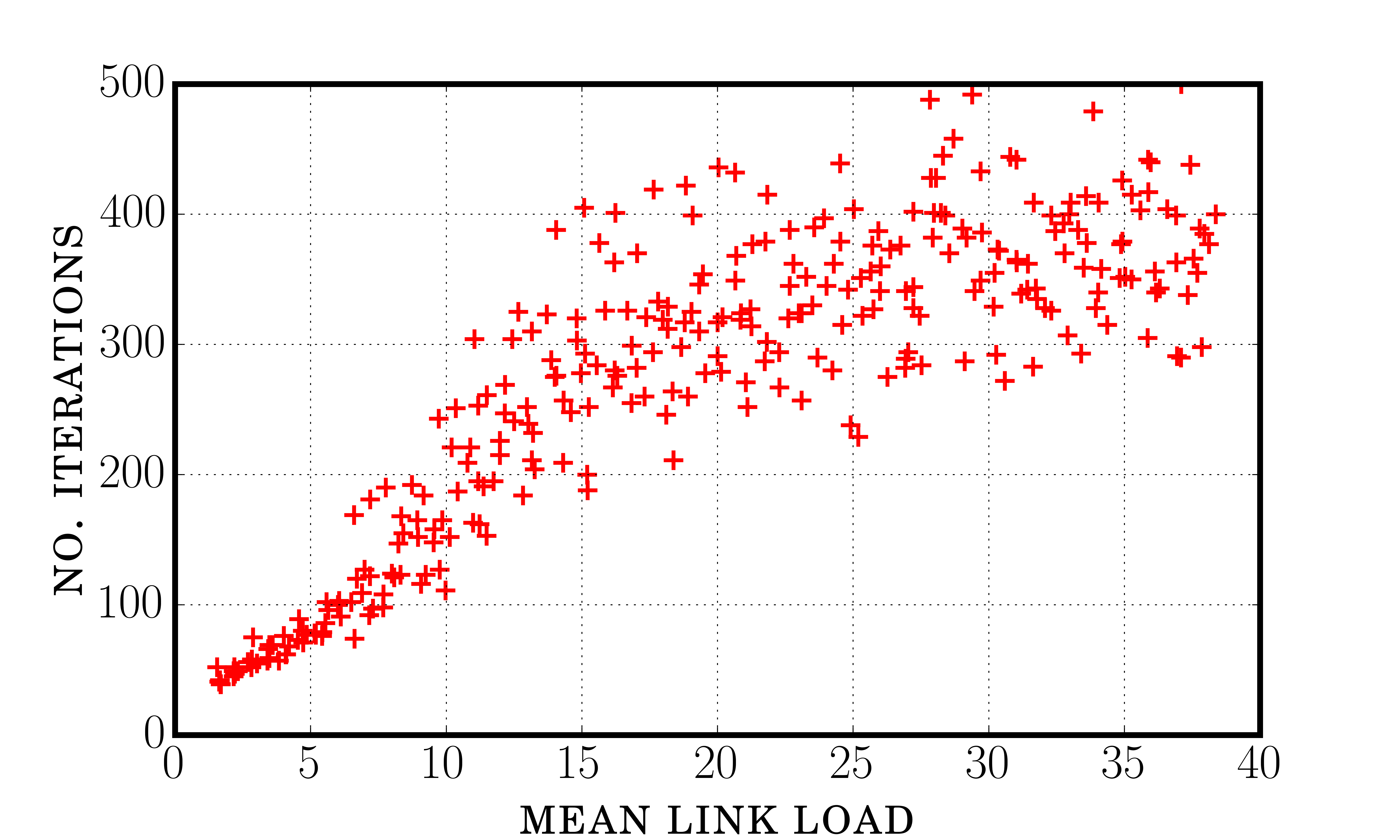}
\caption{Iteration count for FD-ADMM vs the mean link load (average value of the $|R_j|$).}
\label{fig:meanload}
\end{figure}

 \subsection{Algorithm design}
 \label{algodesign}
 Evaluating the alleviation of the compute-intensive parts of C-ADMM was a key concern to motivate and validate the distribution to FD-ADMM. To this aim, we show in Figure~\ref{fig:itercvsfd} the computation time and iteration count for those two algorithms on small instances for a number of requests ranging from 1 to 200. The centralized projection in C-ADMM is executed using the variation of Hildreth's projection algorithm on general polyhedra in \cite{iusem1987simultaneous}. When convergence is desired, a precise stopping criterion for FD-ADMM is available, as the optimality gap can be upper-bounded by the primal and dual residuals, see \cite{boyd2011distributed}. In our case, evaluating those residuals results in computing the absolute variation of two consecutive values of $\tilde{z}$, and the consensus accuracy\footnote{One can choose any other norm in $\bigoplus \RR^{|R_j|}$.} $\max_{r\in R, j\in r}| z_{jr} - \tilde{z}_r|$. This is a first advantage for FD-ADMM implementation as no robust stopping criterion is available for standard gradient descent. When an optimality gap is computed, we thus consider a $10^{-6}$-approximation by FD-ADMM as the reference for all tested algorithms. In Figure~\ref{fig:stepsize}, we illustrated, on a small instance with 200 requests, the number of iterations of FD-ADMM to reach convergence for a various number of the parameter values, in order to evaluate our \emph{adaptive} scheme's accuracy with respect to the empirically \emph{best found} parameter. It shows that our approximation of $\lambda_*$ is fairly satisfactory. In Figure~\ref{fig:itercvsfd}, FD-ADMM shows that distributing the consensus over the links exchanges several more iterations for a reduction of the compute time by two orders of magnitude for small instances. Hence, the distribution does not seem to cost too much convergence rate.  Not surprisingly, the use of a central projection sub-routine makes C-ADMM impossible to scale. The convergence criterion used in Figures~\ref{fig:itercvsfd} and \ref{fig:meanload} is modest ($10^{-1}$). Finally, we plotted a notable behavior of FD-ADMM in Figure~\ref{fig:meanload}. One can imagine a link between the convergence rate and the mean link load, i.e., $\frac{1}{[J|}\sum |R_j|$. This conjecture requires further investigation that we keep for future work.

\subsection{Comparison against Lagrangian method}
\label{responsiveness}

We now compare the proposed FD-ADMM algorithm against the classic LAGR  Algorithm \ref{Kelly}, see \cite{voice2006stability, mccormick2014real,palomar2006tutorial}. To this aim we perform two experiments, in real-time and static scenarios, respectively. 

We start by evaluating the real-time responsiveness of FD-ADMM by considering a small scenario where 200 routes are established and the weights $(w_r^t)_{r\in R, t\in 0\ldots T}$ vary over discrete time $t$, following the formula:
$$w_r^{t+1} \in [(1-a)w_r^t, (1+a)w_r^t]\quad a\in [0,1], $$
where at each event $t$, $w_r^t$ is chosen uniformly within the above interval in which $a$ determines the amplitude of the weight variation. 
In Figure~\ref{fig:optgapVSa} we illustrated the average optimality gap of the two algorithms achieved over 20 events with 10 iterations between each event. We observe that FD-ADMM outperforms LAGR in terms of optimality gap, although the performance of both algorithms is fairly acceptable. However, remarkably, FD-ADMM remains always feasible whereas LAGR constantly violates the constraints as weights $w_r$ change in real-time. Figure~\ref{fig:violated} shows the percentage of constraints of the problem that are violated for each value of the amplitude $a$. 
In fact, LAGR iteratively approaches the fair resource allocation from the outside of the feasible set. This drawback is commonly amended by projecting the solution onto the feasible set. However, this is not doable in our distributed setting, as projection requires costly on-the-fly operations that require full topological information. For such reasons, we claim that the standard LAGR algorithm is not well suited for computing real-time fair allocations in a distributed SDN setting.


 \begin{figure}[t]
 	\centering
 	\includegraphics[width = .45\textwidth]{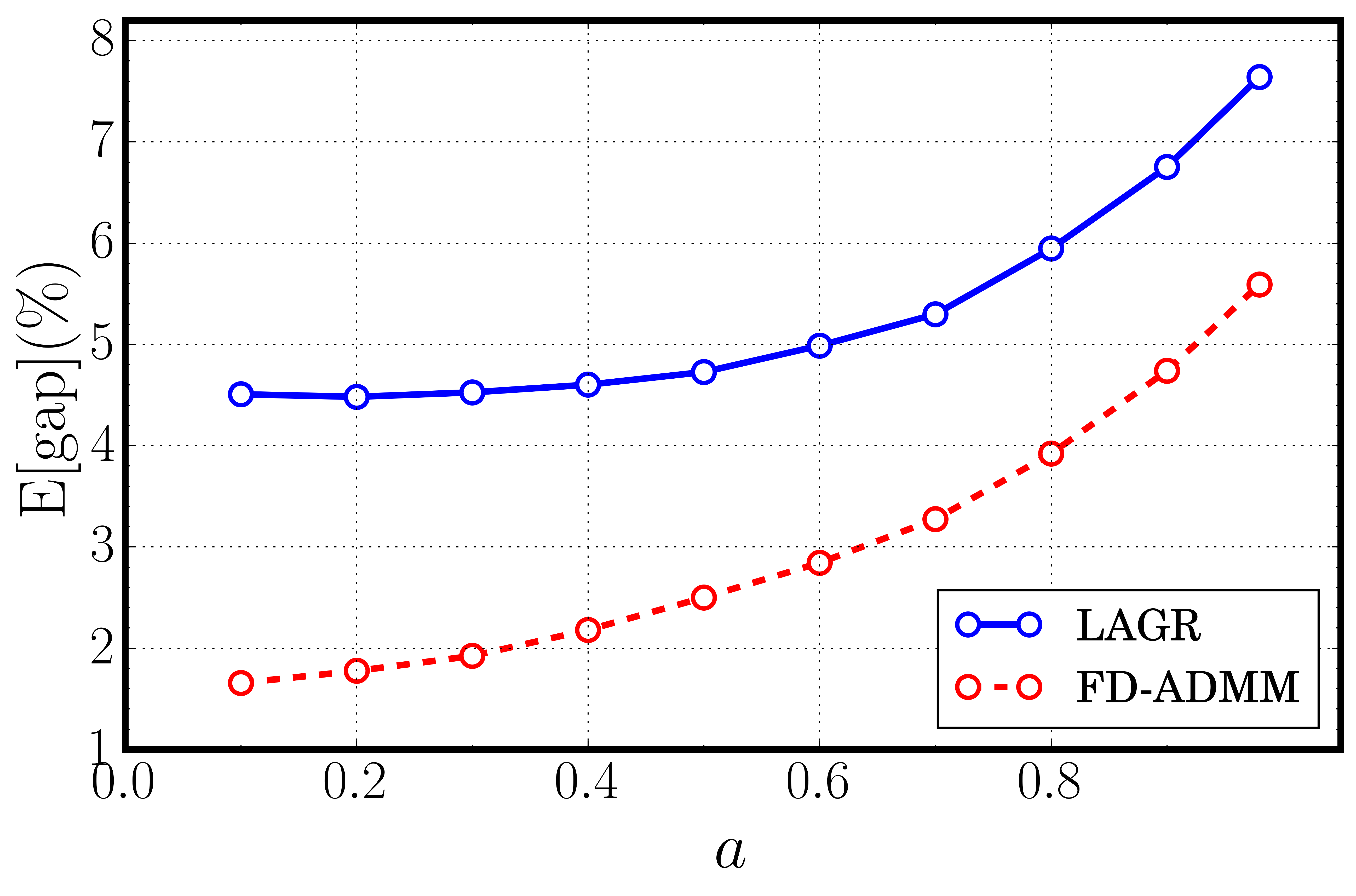}
 	\caption{Average optimality gap $\mathrm{E}[\mathrm{gap}]$ vs. the variation amplitude $a$.}
 	\label{fig:optgapVSa}
 \end{figure}
 
 \begin{figure}[t]
 	\hspace{-4mm}
 	\centering
 	\includegraphics[width = .45\textwidth]{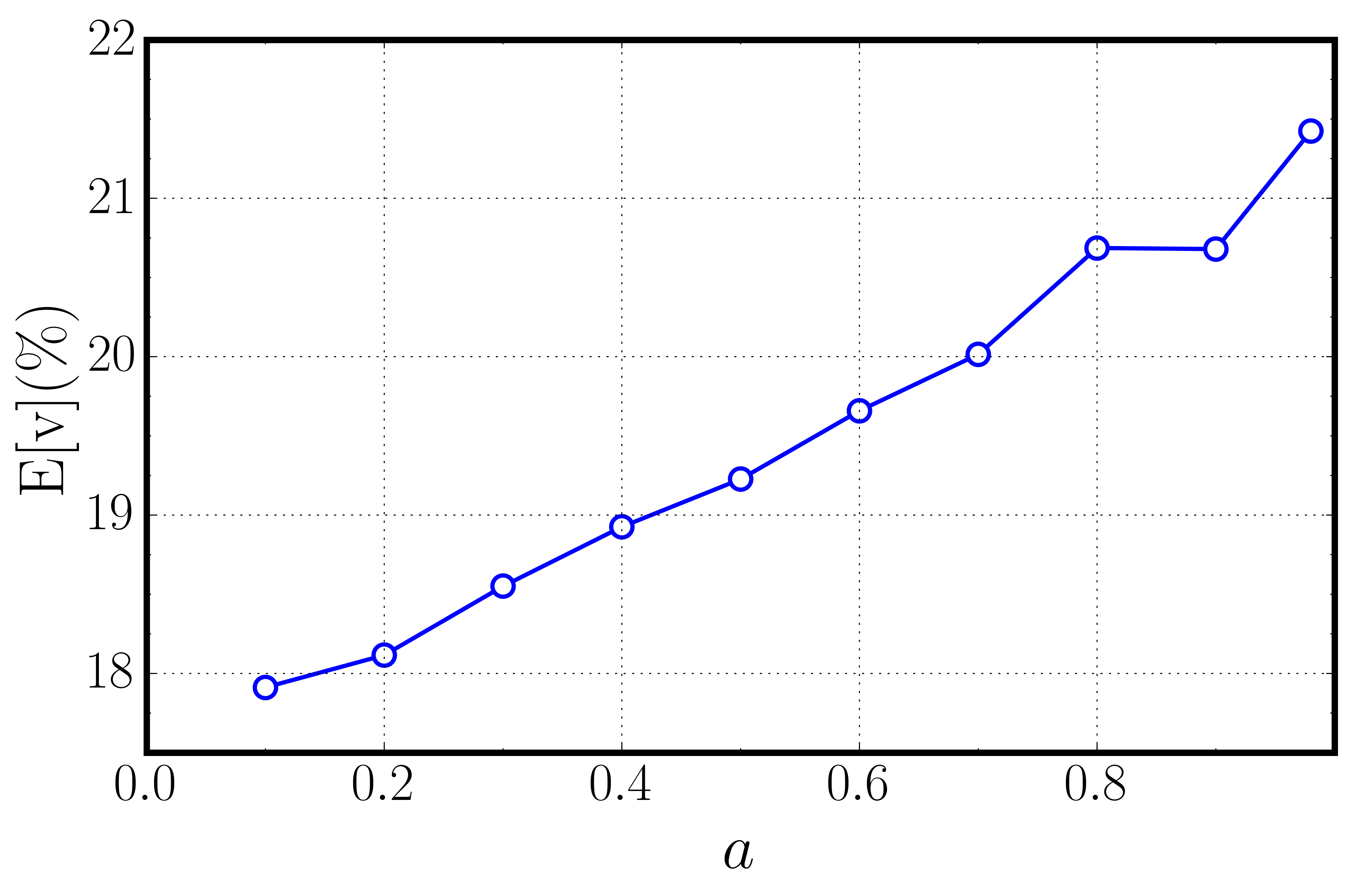}
 	\caption{Average percentage of violated constraints $\mathrm{E}[\mathrm{v}]$ by LAGR vs. the variation amplitude $a$.}
 	\label{fig:violated}
 \end{figure}

In our last experiment we test the two algorithms under a static scenario, where the weights $w_r$ do not vary over time and LAGR has enough time to find at least one feasible solution. In Figure~\ref{fig:realtime} we compare the optimality gap of the \emph{best feasible} solutions found after 5 seconds runtime by FD-ADMM and LAGR, for different instance sizes over BT topology. We observe that FD-ADMM obtains a close-to-optimal feasible solution for all the instance sizes (from 100 to 6000 requests), while LAGR is still far from the optimum especially when the instance becomes large.

To recap, in this section we have demonstrated by experimentation that FD-ADMM reacts quickly to unpredictable network variations, while preserving the feasibility of the solutions computed iteratively. We then claim that FD-ADMM is a good candidate for real-time fair resource allocation in distributed SDN scenarios. 

   \begin{figure}[t]
   	\vspace{-4mm}
   	\centering
   	\includegraphics[width = .5\textwidth]{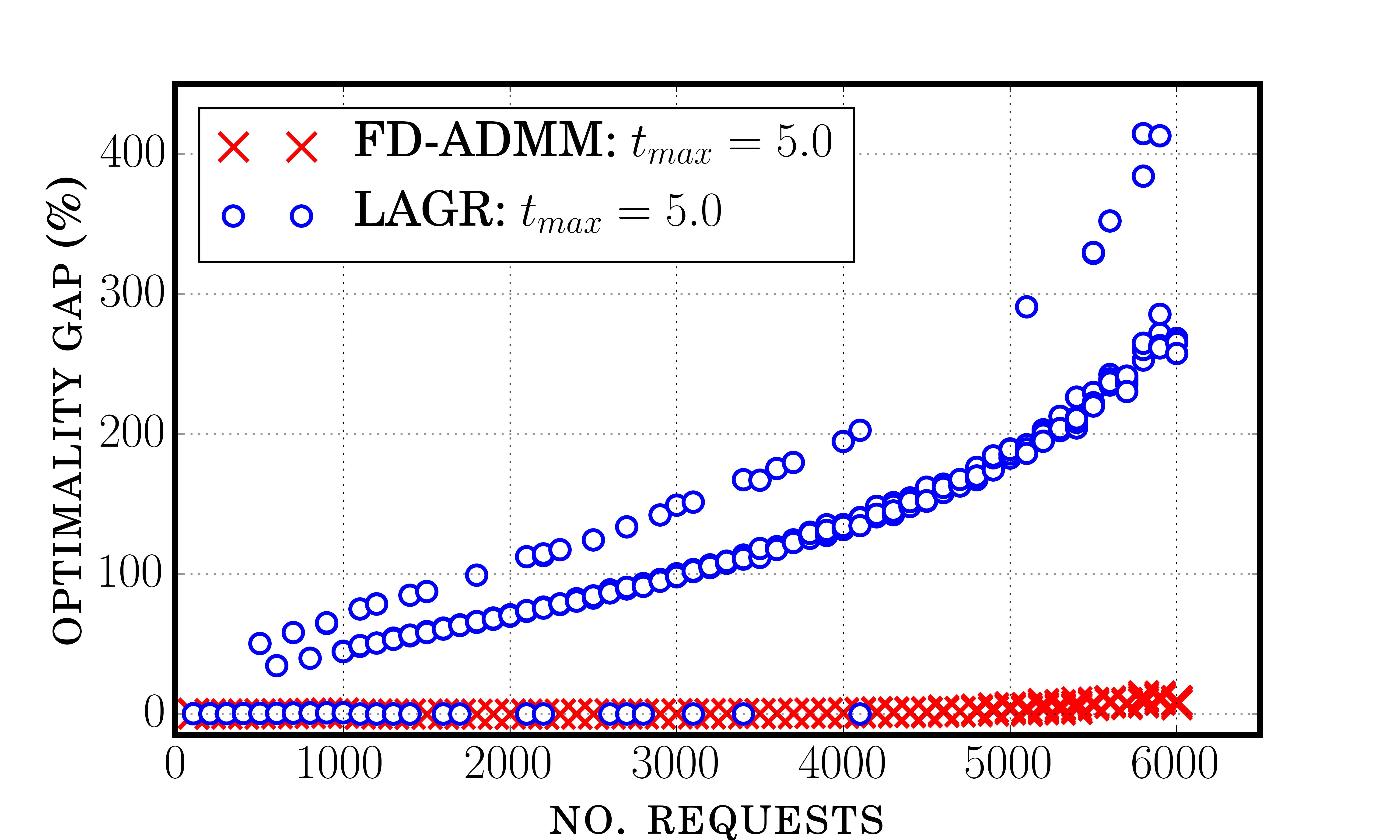}
   	\caption{Optimality gap of the best feasible point found after 5 seconds runtime.}
   	\label{fig:realtime}
   \end{figure}
 
 \section{Conclusions and future work}
  \label{sec:conclusion}
In this paper we addressed the real-time fair resource allocation problem in the context of a distributed SDN control plane architecture. Our main contribution is the design of a distributed algorithm that continuously generates a sequence of feasible solutions and 
adapts to any partitioning of the network into domains. We reformulated the $\alpha$-fair resource allocation problem in the fashion of a general consensus problem to derive the FD-ADMM algorithm. This algorithm can be massively parallelized on several processors that manage different regions of the network, hence fully benefiting from the computing resources of SDN controllers in distributed architectures. We also provided a strategy for a near-optimal estimation of the penalty parameter of FD-ADMM that boosts its convergence. Finally, we compared FD-ADMM to a standard dual Lagrangian decomposition method (LAGR) and we demonstrated how the former is more adapted to a real-time situation where bandwidth has to be adjusted on-the-fly. 
In fact, FD-ADMM ensures a smaller optimality gap since the very first iterations and, most importantly, it produces a feasible fair allocation at all iterations.  

As a next step, we envision to adapt our formulation to the case where multiple candidate paths are available for each request. Moreover, we plan to run FD-ADMM asynchronously while still guarantying near-optimal convergence rate and anytime feasibility.

\bibliography{bibliog}

\begin{thebibliography}{10}

\bibitem{bertsekas1992data}
Dimitri~P Bertsekas, Robert~G Gallager, and Pierre Humblet.
\newblock {\em Data networks}, volume~2.
\newblock Prentice-Hall International Series, 1992.

\bibitem{boyd2011distributed}
Stephen Boyd, Neal Parikh, Eric Chu, Borja Peleato, and Jonathan Eckstein.
\newblock Distributed optimization and statistical learning via the alternating
  direction method of multipliers.
\newblock {\em Foundations and Trends in Machine Learning}, 3(1):1--122, 2011.

\bibitem{charny1995congestion}
Anna Charny, Raj Jain, and David Clark.
\newblock Congestion control with explicit rate indication.
\newblock In {\em {Proc. of IEEE ICC}}, 1995.

\bibitem{chen2011projection}
Yunmei Chen and Xiaojing Ye.
\newblock Projection onto a simplex.
\newblock {\em arXiv preprint arXiv:1101.6081}, 2011.

\bibitem{Deng2016}
Wei Deng and Wotao Yin.
\newblock On the global and linear convergence of the generalized alternating
  direction method of multipliers.
\newblock {\em Journal of Scientific Computing}, 66(3):889--916, 2016.

\bibitem{hassas2012kandoo}
Soheil Hassas~Yeganeh and Yashar Ganjali.
\newblock Kandoo: a framework for efficient and scalable offloading of control
  applications.
\newblock In {\em {Proc. of ACM HotSDN}}, 2012.

\bibitem{he2012douglas}
Bingsheng He and Xiaoming Yuan.
\newblock On the o(1/n) convergence rate of the douglas-rachford alternating
  direction method.
\newblock {\em SIAM J. Numer. Anal.}, 50(2):700--709, April 2012.

\bibitem{iusem1987simultaneous}
Alfredo~N Iusem and Alvaro~R De~Pierro.
\newblock A simultaneous iterative method for computing projections on
  polyhedra.
\newblock {\em SIAM Journal on Control and Optimization}, 25(1):231--243, 1987.

\bibitem{kelly1998rate}
Frank~P Kelly, Aman~K Maulloo, and David~KH Tan.
\newblock Rate control for communication networks: shadow prices, proportional
  fairness and stability.
\newblock {\em Journal of the Operational Research society}, 49(3):237--252,
  1998.

\bibitem{kreutz2015software}
Diego Kreutz, Fernando~MV Ramos, Paulo~Esteves Verissimo, Christian~Esteve
  Rothenberg, Siamak Azodolmolky, and Steve Uhlig.
\newblock Software-defined networking: A comprehensive survey.
\newblock {\em Proc. of the IEEE}, 103(1):14--76, 2015.

\bibitem{marasevic2016fast}
Jelena Marasevic, Clifford Stein, and Gil Zussman.
\newblock A fast distributed stateless algorithm for alpha-fair packing
  problems.
\newblock In {\em Proc. of ICALP}, volume~55, pages 54--1, 2016.

\bibitem{mccormick2014real}
Bill McCormick, Frank Kelly, Patrice Plante, Paul Gunning, and Peter
  Ashwood-Smith.
\newblock Real time alpha-fairness based traffic engineering.
\newblock In {\em Proc. of ACM HotSDN}, pages 199--200, 2014.

\bibitem{mo2000fair}
Jeonghoon Mo and Jean Walrand.
\newblock Fair end-to-end window-based congestion control.
\newblock {\em IEEE/ACM Transactions on Networking (ToN)}, 8(5):556--567, 2000.

\bibitem{mota2012distributed}
Jo{\~a}o~FC Mota, Jo{\~a}o~MF Xavier, Pedro~MQ Aguiar, and Markus Puschel.
\newblock Distributed admm for model predictive control and congestion control.
\newblock In {\em Proc. of IEEE CDC}, 2012.

\bibitem{dhody2009autobandwidth}
Udayasree Palle, Dhruv Dhody, Ravi Singh, Luyuan Fang, and Rakesh Gandhi.
\newblock {PCEP Extensions for MPLS-TE LSP Automatic Bandwidth Adjustment with
  Stateful PCE}.
\newblock Internet-Draft draft-dhody-pce-stateful-pce-auto-bandwidth-09,
  Internet Engineering Task Force, November 2016.
\newblock Work in Progress.

\bibitem{palomar2006tutorial}
Daniel~P{\'e}rez Palomar and Mung Chiang.
\newblock A tutorial on decomposition methods for network utility maximization.
\newblock {\em IEEE Journal on Selected Areas in Communications},
  24(8):1439--1451, 2006.

\bibitem{skivee2004distributed}
Fabian Skiv{\'e}e and Guy Leduc.
\newblock {A distributed algorithm for weighted max-min fairness in MPLS
  networks}.
\newblock In {\em International Conference on Telecommunications}, pages
  644--653. Springer, 2004.

\bibitem{stallings2013openflow}
William Stallings.
\newblock Software-defined networks and openflow.
\newblock {\em The internet protocol Journal}, 16(1):2--14, 2013.

\bibitem{sundaresan2016iterative}
Rajesh Sundaresan et~al.
\newblock An iterative interior point network utility maximization algorithm.
\newblock {\em arXiv preprint arXiv:1609.03194}, 2016.

\bibitem{vaughan2011openflow}
Steven~J Vaughan-Nichols.
\newblock Openflow: The next generation of the network?
\newblock {\em Computer}, 44(8):13--15, 2011.

\bibitem{voice2006stability}
Thomas Voice.
\newblock Stability of multi-path dual congestion control algorithms.
\newblock In {\em Proc. of Valuetools}, page~56. ACM, 2006.

\end{thebibliography}
\bibliographystyle{plain}

\newpage
\tableofcontents
\end{document}